\documentclass[11pt]{article}

\usepackage[margin=0.9in]{geometry}
\usepackage{amsmath,amssymb,amsthm,mathtools}
\usepackage{booktabs}
\usepackage{microtype}
\usepackage[hidelinks]{hyperref}
\usepackage{xurl}
\newtheorem{theorem}{Theorem}
\newtheorem{lemma}[theorem]{Lemma}

\theoremstyle{remark}
\newtheorem{remark}[theorem]{Remark}

\newcommand{\R}{\mathbb R}
\newcommand{\Z}{\mathbb Z}

\newcommand{\OPT}{\operatorname{OPT}}
\newcommand{\ALG}{\operatorname{ALG}}
\newcommand{\CPB}{C_{\mathrm{PB}}}
\newcommand{\CRdet}{\operatorname{CR}_{\mathrm{det}}}

\hypersetup{
  pdftitle={A Deterministic Constant-Competitive Algorithm for Dynamic Mixture-of-Experts Serving},
  pdfauthor={Ian D'Ambrosio}
}

\title{A Deterministic Constant-Competitive Algorithm for\\
Dynamic Mixture-of-Experts Serving}
\author{Ian D'Ambrosio\\
\small Nth Research Collective}
\date{26 August 2026}

\begin{document}
\maketitle

\begin{abstract}
Dynamic Mixture-of-Experts Serving allocates $k$ replica GPUs among $m$
experts as workloads change.  At each round, the online algorithm sees the
current workload, chooses integral replica counts, and pays bottleneck service
cost plus replica movement.  It does not know future workloads.  Huang, Lou,
and Xiao gave an \mbox{$O(\sqrt{\log k})$}-competitive randomized algorithm for
this problem.  We prove a deterministic $O(1)$-competitive algorithm.  For
every number of experts and every $k\geq1$, the algorithm satisfies
\[
 \ALG_{\rm det}\leq10\CPB\cdot\OPT+(5\CPB+8)k+16,
\]
where $\CPB$ is the absolute constant from Chasing Positive Bodies at resource
augmentation one and covering sparsity two.  Consequently, the deterministic
competitive-ratio upper bound is
\[
 \boxed{\CRdet(k)\leq10\CPB\quad\text{for every }k\geq1},
\]
which is $O(1)$ because $\CPB$ is absolute.  The routine universal lower bound
then gives $\CRdet(k)=\Theta(1)$.  The multiplicative factor does not depend on
the number of experts, the replica budget, the horizon, or the workload
values.  Thus randomization is not needed for the asymptotic guarantee.

The proof has two layers.  A finite tangent envelope, summable positive resets,
and a nonexpansive balanced projection reduce reciprocal-max service costs to a
deterministic exact-budget fractional path.  A new deterministic rounding
theorem converts every such path to integral allocations with service
distortion three and movement bounded by the fractional movement plus $6k$.
The complete reduction, rounding theorem, causal composition, and quantified
main theorem are machine-checked in Lean 4 relative to the positive-body result
as the sole scientific source premise.  The theorem concerns the allocation
model above.  It does not include network topology, shared-edge congestion, or
routing decisions.
\end{abstract}

\section{Introduction}

Dynamic Mixture-of-Experts (MoE) inference must decide how many replica GPUs to
assign to each expert as workloads change.  Moving replicas has a switching
cost, while insufficient replication increases bottleneck latency.  Huang,
Lou, and Xiao formalized this tradeoff as an online problem and proved an
\mbox{$O(\sqrt{\log k})$} randomized competitive ratio for the integral primal
problem~\cite{HLX2026}.  Their \mbox{$\Omega(\sqrt{\log k})$} construction concerns an
offset-Fenchel-dual maximization problem.  Weak duality points in the wrong
direction for transferring it to a primal lower bound, and the asymptotic
primal ratio remained open.

Two contemporaneous follow-ups then gave constant randomized bounds for the
primal problem.  Vergeres proved a deterministic factor-$16$ bound for the
equality-budget fractional relaxation and, using Lazy Threshold Rounding, a
randomized factor-$48$ integral bound against an oblivious
adversary~\cite{Vergeres2026}.  An earlier version of this manuscript obtained a
constant randomized integral ratio through the positive-body reduction below
and the same rounding theorem~\cite{HLX2026}.  Neither result gave a
deterministic integral bound.  A simple deterministic baseline keeps the
supplied integral allocation fixed.  Since every fractional coordinate lies in
$[0,k]$, its service is within a factor $k+1$ of that fractional path.  Thus the
previous deterministic bracket was $\Omega(1)$ to $O(k)$.

This paper closes that deterministic gap: for every $k\geq1$,
$\CRdet(k)\leq10\CPB$, an $O(1)$ upper bound because $\CPB$ is absolute.  The
multiplicative factor is independent of the number of experts, replica budget,
initial state, horizon, and workload values.

The proof preserves the positive-body reduction from the previous version.
The epigraph of $r/(1+u)$ admits a finite positive-polyhedral inner
approximation with constant distortion.  This permits a reduction to the
positive-body chasing theorem of Bhattacharya, Buchbinder, Levin, and
Saranurak~\cite{BBLS2023}.  A finite tangent grid handles the continuum of
reciprocal tangents, geometrically shrinking positive resets charge service to
movement, and a balanced projection removes resource augmentation without
expanding movement.

The new component is deterministic pathwise rounding.  A coordinate receives
an integral unit only after its fractional mass crosses one of the activation
thresholds at $2,5,8,\ldots$.  Once present, the unit remains until it has at
least one unit of certified depth.  The top-unit potential
\[
 \Phi(y,z)=\sum_i(3y_i-1-z_i)_+
\]
decreases by at least one at every repair transfer.  Exact fractional budgets
bound its increase between rounds by half the fractional $\ell_1$ movement.
The resulting telescope has service distortion three and movement stretch one
plus an additive $6k$, independent of the expert count and horizon.

\paragraph{Contribution.}
Our contribution is the deterministic $O(1)$ upper bound and the pathwise
rounding theorem that makes it possible.  This improves the prior deterministic
$O(k)$ baseline to $O(1)$ and shows that randomization is not asymptotically
necessary.  The factor-one lower bound is only the universal observation that
every online path is a feasible offline path; it is not a scientific
contribution.

\paragraph{Verification.}
The complete theorem is formalized in Lean 4 relative to the two-sparse
positive-body chaser as its sole scientific source premise.  The companion
archive rebuilds the proof and regenerates the exact controls from one command,
including a deliberately false rounder without hysteresis.

\section{Model and main result}

Fix $m\geq1$ experts and $k\geq1$ replica GPUs.  The integral state space is
\[
 X_{m,k}=\left\{x\in\Z_{\geq0}^{m}:\sum_{i=1}^{m}x_i=k\right\}.
\]
At round $t$, a nonnegative workload vector $r_t\in\R_{\geq0}^{m}$ is
revealed.  The online algorithm chooses $x_t\in X_{m,k}$ and pays
\[
 f_t(x_t)+\lVert x_t-x_{t-1}\rVert_1,
 \qquad
 f_t(x)=\max_{i\in[m]}\frac{r_{t,i}}{1+x_i},
\]
where the integral initial state $x_0$ is part of the instance.  For a fixed
finite workload sequence, $\OPT$ is the minimum total cost over all integral
paths.

For fixed $k$, write $\CRdet(k)$ for the infimum of factors $\Gamma$ with the
following property: for every $m\geq1$ there is a constant $A_{m,k}\geq0$ such
that, for every integral initial state, some deterministic causal integral
algorithm satisfies
$\ALG\leq\Gamma\cdot\OPT+A_{m,k}$ on every finite workload sequence.  The
additive term may depend on $m$ and $k$, but not on the horizon or workloads.
Causality means that equal request prefixes produce equal allocation prefixes.

Let $\CPB$ denote an absolute competitive constant supplied by Theorem~1.1 of
\cite{BBLS2023} at resource augmentation $\epsilon=1$ and maximum covering-row
support $d=2$.  We absorb into $\CPB$ the constant from reducing a positive body
to positive halfspaces and the standard conversion between upward movement and
full $\ell_1$ movement.

\paragraph{Two properties of the service cost.}
The proof uses two elementary features of $f_t$.  First, its reciprocal terms
are smooth under capacity scaling: for every $r,u\geq0$,
\[
 \frac{r}{1+u/2}\leq2\frac{r}{1+u}.
\]
Thus halving capacity from $2k$ to $k$ while retaining at least half of every
coordinate increases service by at most a factor two.  Second, $f_t$ is the
maximum of univariate coordinate functions.  Its epigraph therefore separates
into constraints involving only the height and one allocation coordinate, and
the tangent discretization preserves covering-row support $d=2$.
Equation~\eqref{eq:cover-row} makes the sparsity explicit, while
\eqref{eq:service-height} applies the capacity-scaling inequality.

\begin{theorem}[Deterministic $O(1)$-competitive upper bound]\label{thm:upper}
For every $m,k\geq1$ and every integral initial state, there is a causal
deterministic integral online algorithm such that every finite nonnegative
workload sequence satisfies
\[
 \ALG_{\rm det}
 \leq 10\CPB\cdot\OPT+(5\CPB+8)k+16.
\]
The multiplicative constant is independent of $m$, $k$, the initial state, the
horizon, and the workloads.
\end{theorem}

\paragraph{Interpretation and scope.}
The algorithm sees each workload before choosing that round's allocation, but
it never sees future workloads.  Its total service and movement cost is at most
a fixed multiple of the best offline integral path, plus the displayed
additive term.  The fixed multiple is independent of every problem dimension.
This is a competitive-cost guarantee, not a running-time or serving-latency
benchmark.  The model allocates replicas but does not model physical network
topology, shared links, or routing.

Since $\CPB$ is absolute, Theorem~\ref{thm:upper} gives the promised upper bound
\[
 \boxed{\CRdet(k)\leq10\CPB\quad\text{for every }k\geq1}.
\]
In particular, $\CRdet(k)=O(1)$.  Every realized online path is also feasible
offline.  Scaling one positive one-round workload makes any
horizon-independent additive negligible, so $\CRdet(k)\geq1$.  Consequently
$\CRdet(k)=\Theta(1)$.  The lower inequality is routine and is not a separate
contribution.  Here $O(1)$ refers to the competitive factor multiplying
$\OPT$, not to the wall-clock running time of the algorithm.

\section{Source theorem}

We use the following exact consequence of a published result.

\paragraph{Positive-body chasing.}
A positive body has the form
\[
 K_t=\{z\in\R_{\geq0}^{D}:C^tz\geq\mathbf1,\;P^tz\leq\mathbf1\},
\]
where both matrices are nonnegative.  Theorem~1.1 of~\cite{BBLS2023} gives,
for $\epsilon\in(0,1]$, an
$O(\epsilon^{-1}\log(d/\epsilon))$-competitive online chaser whose covering
constraints remain exact and whose packing constraints may be violated by a
factor $1+\epsilon$; $d$ is the largest covering-row support.  Its movement is
compared with an exact offline path beginning at zero.  Thus, when
$\epsilon=1$ and $d=2$, the movement factor $\CPB$ is absolute, covering rows
remain exact, and packing right-hand sides may be doubled.  Strictly positive
right-hand sides normalize to one without changing row support.  A covering
row with zero coefficients and zero right-hand side is vacuous and may be
deleted.

The previous version also invoked Lazy Threshold Rounding from~\cite{HLX2026}.
The deterministic rounding theorem in Section~\ref{sec:det-rounding} replaces
that step.  No randomized rounding theorem is a premise of
Theorem~\ref{thm:upper}.

\section{A finite positive approximation to reciprocal service}

Fix a workload coordinate $r\geq0$ and augmented allocation $u\geq0$.  Write
$q=1+u$.  For a scale $p\geq1$, define the tangent
\[
 L_{r,p}(u)=\frac{r(2p-q)}{p^2}
           =\frac{r(2p-1-u)}{p^2}.
\]

\begin{lemma}[Tangent envelope]\label{lem:tangent}
For all $r\geq0$, $q>0$, and $p>0$,
\[
 \frac rq-L_{r,p}(u)=\frac{r(q-p)^2}{qp^2}\geq0.
\]
If $0\leq u\leq2k$, some $p\in\{1,\ldots,2k+1\}$ also satisfies
\[
 L_{r,p}(u)\geq\frac34\frac r{1+u}.
\]
\end{lemma}

\begin{proof}
The first identity follows by putting the two terms over $qp^2$.  For the
second claim, take $p=\lceil q\rceil$.  Then $1\leq p\leq2k+1$ and
$q\leq p<q+1\leq2q$.  Hence $a=q/p\in(1/2,1]$, and
\[
 \frac{L_{r,p}(u)}{r/q}=a(2-a)\geq\frac34.
\]
The case $r=0$ is immediate.
\end{proof}

Define the finite envelope
\[
 H_r(u)=\max_{1\leq p\leq2k+1}L_{r,p}(u).
\]
Lemma~\ref{lem:tangent} gives
\begin{equation}\label{eq:envelope}
 \frac34\frac r{1+u}\leq H_r(u)\leq\frac r{1+u}
 \qquad(0\leq u\leq2k).
\end{equation}
For $r>0$, the inequality $s\geq L_{r,p}(u)$ is equivalent to the positive
covering row
\begin{equation}\label{eq:cover-row}
 \frac{p^2}{r(2p-1)}s+\frac1{2p-1}u\geq1.
\end{equation}
Its support has size at most two.  This is the $d=2$ sparsity supplied by the
maximum-of-univariate form of $f_t$.  Coordinates with $r=0$ require no row;
equivalently, one may insert and delete a vacuous zero-coefficient covering row
with zero right-hand side, as in the formal reduction.

\begin{remark}
A geometric grid $p\in\{(3/2)^j:(3/2)^j\leq1+2k\}$ gives the same $3/4$
factor with $O(\log(k+1))$ rows per expert.  The concrete Lean reduction uses
the elementary integer grid above; both envelope lemmas are machine-checked.
\end{remark}

\section{The event/reset reduction}

Use augmented coordinates $(u_1,\ldots,u_m,s)\in\R_{\geq0}^{m+1}$.  When
request $r_t$ arrives, present the positive body
\begin{equation}\label{eq:event}
 E_t=\left\{(u,s):\sum_i u_i\leq k,\quad
 s\geq L_{r_{t,i},p}(u_i)\;\text{for all }i,p\right\}.
\end{equation}
After taking the round-$t$ MoE action, present the reset body
\begin{equation}\label{eq:reset}
 R_t=\left\{(u,s):\sum_i u_i\leq k,\quad s\leq\delta_t\right\},
 \qquad \delta_t=2^{-t}\quad(t\geq1),
\end{equation}
and set $\delta_0:=0$ for the initial accounting convention.
Both packing rows have positive right-hand side after normalization.  The event
covering rows have support at most two.  Both bodies are nonempty.  The reset is
processed after the action at time $t$ and before the next request, so the
construction remains causal.

Run the positive-body chaser with $\epsilon=1$.  At an event, write its point
as $(u_t,s_t)$.  At the following reset, write its height as $b_t$.  Then
\begin{equation}\label{eq:aug-feasible}
 \sum_i u_{t,i}\leq2k,\qquad s_t\geq H_{r_{t,i}}(u_{t,i})\quad\forall i,
 \qquad 0\leq b_t\leq2\delta_t.
\end{equation}

\subsection{Removing resource augmentation}

Map every $u\geq0$ with $\sum_i u_i\leq2k$ to the exact-budget fractional
allocation
\begin{equation}\label{eq:balanced}
 \rho(u)=\frac{k-\frac12\sum_i u_i}{m},
 \qquad x_i(u)=\frac{u_i}{2}+\rho(u).
\end{equation}

\begin{lemma}[Balanced projection]\label{lem:balanced}
The map in \eqref{eq:balanced} satisfies $x(u)\geq0$,
$\sum_i x_i(u)=k$, $x_i(u)\geq u_i/2$, and
\[
 \lVert x(u)-x(v)\rVert_1\leq\lVert u-v\rVert_1.
\]
\end{lemma}

\begin{proof}
The budget bound gives $\rho(u)\geq0$, and summing
\eqref{eq:balanced} gives total mass $k$.  Moreover,
\begin{align*}
 \lVert x(u)-x(v)\rVert_1
 &\leq\frac12\lVert u-v\rVert_1+m|\rho(u)-\rho(v)|\\
 &=\frac12\lVert u-v\rVert_1
   +\frac12\left|\sum_i(u_i-v_i)\right|\\
 &\leq\lVert u-v\rVert_1.
\end{align*}
\end{proof}

Combining \eqref{eq:envelope}, \eqref{eq:aug-feasible}, and the capacity-halving
bound $1+x_i(u_t)\geq(1+u_{t,i})/2$ gives
\begin{equation}\label{eq:service-height}
 f_t(x(u_t))
 \leq2\max_i\frac{r_{t,i}}{1+u_{t,i}}
 \leq\frac83s_t.
\end{equation}

\subsection{Charging service to vertical movement}

Let $a_{t-1}$ be the chaser height after the preceding reset, with $a_0=0$.
Thus $a_{t-1}\leq2\delta_{t-1}$ for every $t\geq1$, including the initial
case.  Let $V_t$ be the vertical movement while processing $E_t,R_t$.
Endpoint distances imply
\[
 V_t\geq|s_t-a_{t-1}|+|s_t-b_t|
     \geq2s_t-a_{t-1}-b_t.
\]
Using \eqref{eq:service-height}, $a_{t-1}\leq2\delta_{t-1}$, and
$b_t\leq2\delta_t$, we obtain
\[
 f_t(x(u_t))
 \leq\frac43V_t+\frac83(\delta_{t-1}+\delta_t).
\]
Because $\delta_0=0$ and $\sum_{t\geq1}\delta_t=1$, summing the error terms
contributes at most $16/3$.  If $M$ denotes the chaser's complete movement,
then
\begin{equation}\label{eq:service-bridge}
 \sum_t f_t(x(u_t))\leq\frac43M+\frac{16}{3}.
\end{equation}
The balanced projection is nonexpansive.  Write $x_t:=x(u_t)$ for $t\geq1$
and identify the supplied integral $x_0$ with its real vector.  The first state
can be at distance at most $2k$ from $x_0$, the diameter of the exact simplex.
Hence
\begin{equation}\label{eq:movement-bridge}
 \lVert x_1-x_0\rVert_1+
   \sum_{t=2}^{T}\lVert x_t-x_{t-1}\rVert_1\leq M+2k.
\end{equation}

\section{Comparison with the MoE optimum}

Take any integral comparator path $y_1,\ldots,y_T$ and let
$h_t=f_t(y_t)$.  At event time use the body point $(y_t,h_t)$, and at reset
time use $(y_t,0)$.  The tangent-underestimator identity makes the event point
feasible; the reset point is feasible as well.  Starting from zero, this lifted
body path has movement
\begin{align}
 &k+\sum_{t=2}^{T}\lVert y_t-y_{t-1}\rVert_1+2\sum_{t=1}^{T}h_t
 \notag\\
 &\hspace{35mm}\leq k+2\,\operatorname{cost}(y_1,\ldots,y_T).
 \label{eq:comparator}
\end{align}
Applying the positive-body theorem to an optimal MoE path gives
\begin{equation}\label{eq:M}
 M\leq\CPB(k+2\OPT).
\end{equation}

\section{Deterministic pathwise rounding}\label{sec:det-rounding}

This section supplies the new ingredient that replaces Lazy Threshold
Rounding: a deterministic, pathwise conversion from fractional to integral
allocations.  Its service factor is three, while its total movement is at most
the fractional movement plus $6k$.  For $a\in\R$, write
$a_+=\max\{0,a\}$.  Given
$z\in\R_{\geq0}^m$ with $\sum_i z_i=k$, define the required integral level
\begin{equation}\label{eq:required}
 \ell_i(z)=\max\left\{0,\left\lceil\frac{z_i-2}{3}\right\rceil\right\}.
\end{equation}
It is the smallest nonnegative integer satisfying
\begin{equation}\label{eq:det-service-threshold}
 1+z_i\leq3(1+\ell_i(z)).
\end{equation}
Since $\ell_i(z)\leq z_i$, we have $\sum_i\ell_i(z)\leq k$.

\paragraph{Repair rule.}
Given the previous integral state $y\in X_{m,k}$ and current fractional state
$z$, choose the least index $i$ with $y_i<\ell_i(z)$ and temporarily add one
unit there.  Call the resulting mass-$(k+1)$ state $w$.  Among its positive
coordinates choose, using a fixed tie-breaking rule, a coordinate $j$
maximizing
\begin{equation}\label{eq:top-depth}
 d_j(w,z)=3w_j-1-z_j.
\end{equation}
Set $y:=w-e_j$ and repeat until every coordinate meets its required level.
Starting from the supplied $y_0$, apply this rule to $z_t$ to obtain $y_t$.

\begin{lemma}[Deep donor]\label{lem:deep-donor}
Every nonnegative integral vector $y$ of mass $k+1$ contains a positive
coordinate $j$ with $3y_j-1-z_j\geq1$.  Removing one unit from such a
coordinate leaves it at or above $\ell_j(z)$.
\end{lemma}

\begin{proof}
Let $S=\{i:y_i>0\}$.  If every $j\in S$ had $3y_j-1-z_j<1$, then
\[
 \sum_{j\in S}z_j
   >\sum_{j\in S}(3y_j-2)
   =3(k+1)-2|S|.
\]
Because $|S|\leq k+1$, the right side is at least $k+1$, while
$\sum_{j\in S}z_j\leq\sum_jz_j=k$, a contradiction.  The depth inequality
rearranges to $1+z_j\leq3y_j=3(1+y_j-1)$, so
\eqref{eq:det-service-threshold} implies $\ell_j(z)\leq y_j-1$.
\end{proof}

The receiver cannot be the donor.  Deficiency before the addition implies
$z_i>3y_i+2$, so its new depth is negative.  Lemma~\ref{lem:deep-donor}
therefore makes every transfer feasible.  The receiver deficiency falls by one
and the donor remains nondeficient, so the loop terminates.

Define the top-unit potential
\begin{equation}\label{eq:top-potential}
 \Phi(y,z)=\sum_i(3y_i-1-z_i)_+.
\end{equation}

\begin{lemma}[Repair charge]\label{lem:repair-charge}
If $R_z(y)$ is the repaired state, then
\[
 \lVert R_z(y)-y\rVert_1
 \leq2\bigl(\Phi(y,z)-\Phi(R_z(y),z)\bigr).
\]
\end{lemma}

\begin{proof}
Adding a required unit leaves the receiver's potential unchanged because the
top-depth expression is negative before and immediately after the addition.
Removing a donor of depth $d\geq1$ changes its contribution from $d$ to
$(d-3)_+$, a decrease of at least one.  Each transfer changes two coordinates
by one and therefore contributes at most two to $\ell_1$ movement.  The
triangle inequality then bounds the final displacement by the transfer path.
Summing over transfers proves the claim.
\end{proof}

\begin{lemma}[External variation]\label{lem:external-variation}
For a fixed integral $y$ and exact-budget fractional states $z,z'$,
\[
 2\bigl(\Phi(y,z')-\Phi(y,z)\bigr)
 \leq\lVert z'-z\rVert_1.
\]
\end{lemma}

\begin{proof}
Coordinatewise,
\[
 (3y_i-1-z_i')_+-(3y_i-1-z_i)_+\leq(z_i-z_i')_+.
\]
Equal budgets give
$\sum_i(z_i-z_i')_+=\frac12\lVert z'-z\rVert_1$.  Summing proves the result.
\end{proof}

\begin{theorem}[Deterministic simplex rounding]\label{thm:det-rounding}
For every finite exact-budget fractional path $z_1,\ldots,z_T$ and every
supplied $y_0\in X_{m,k}$, the repair rule produces a deterministic
prefix-causal integral path $y_1,\ldots,y_T$ satisfying
\begin{align}
 1+z_{t,i}&\leq3(1+y_{t,i})&&\text{for every $t,i$},
   \label{eq:det-coordinate}\\
 \sum_{t=1}^T\lVert y_t-y_{t-1}\rVert_1
 &\leq6k+\sum_{t=2}^T\lVert z_t-z_{t-1}\rVert_1.
   \label{eq:det-movement}
\end{align}
Consequently, every nonnegative workload sequence satisfies
\begin{equation}\label{eq:det-service}
 \sum_t f_t(y_t)\leq3\sum_t f_t(z_t).
\end{equation}
\end{theorem}

\begin{proof}
Termination and exact budget follow from the repair rule and
Lemma~\ref{lem:deep-donor}.  At termination
$y_{t,i}\geq\ell_i(z_t)$, which gives \eqref{eq:det-coordinate} and
\eqref{eq:det-service}.

For movement, apply Lemma~\ref{lem:repair-charge} at every round.  For each
$t\geq2$, split at $\Phi(y_{t-1},z_{t-1})$.  The repair decreases telescope,
and Lemma~\ref{lem:external-variation} bounds every external increase.  The
final potential is nonnegative, while
\[
 \Phi(y_0,z_1)\leq\sum_i3y_{0,i}=3k.
\]
Therefore
\[
 \frac12\sum_{t=1}^T\lVert y_t-y_{t-1}\rVert_1
 \leq3k+\frac12\sum_{t=2}^T\lVert z_t-z_{t-1}\rVert_1,
\]
which is \eqref{eq:det-movement}.  Each round uses only $y_{t-1}$ and $z_t$,
so the path is prefix-causal.
\end{proof}

\begin{remark}
Fixed tie-breaking makes the mathematical repair algorithm constructive.  The
Lean development instead uses \texttt{Classical.choose} from a proved nonempty
one-step repair relation.  The choice is state-local and preserves the same
prefix law; the formal theorem makes no runtime claim for the complete
positive-body algorithm.
\end{remark}

\section{Composition}

\begin{proof}[Proof of Theorem~\ref{thm:upper}]
Apply Theorem~\ref{thm:det-rounding} to the deterministic fractional path
constructed above.  Its internal fractional movement is no larger than the
left side of \eqref{eq:movement-bridge}.  Equations
\eqref{eq:service-bridge}, \eqref{eq:movement-bridge}, and
\eqref{eq:det-service} therefore give
\begin{align*}
 \ALG_{\rm det}
 &\leq(M+2k)+6k+3\left(\frac43M+\frac{16}{3}\right)\\
 &=5M+8k+16\\
 &\leq10\CPB\cdot\OPT+(5\CPB+8)k+16,
\end{align*}
where the last step uses \eqref{eq:M}.  Every component is deterministic and
causal, and the additive is independent of the expert count, initial state,
horizon, and workloads.  This proves the $O(1)$ competitive-ratio upper bound
$\CRdet(k)\leq10\CPB$.
\end{proof}

\section{Formal and executable verification}

The formal development uses Lean 4.32.2 and pinned Mathlib 4.32.2.  Its main
declarations are
\begin{itemize}
\item
\texttt{DynamicMoeDeterministicMain.}\newline
\texttt{dynamicMoe\_explicit\_deterministic\_upper\_of\_positive\_body\_source};
and
\item
\texttt{DynamicMoeDeterministicMain.}\newline
\texttt{dynamicMoe\_deterministic\_theta\_one\_of\_positive\_body\_source}.
\end{itemize}
They state the quantifiers in the order
\[
 \forall k\geq1,\ \forall m\geq1,\ \exists A\geq0,\ \forall x_0,
 \ \exists\text{ one deterministic causal algorithm},\ \forall(r_t).
\]
The explicit additive is $(5\CPB+8)k+16$.  Lean reports only the standard axioms
\texttt{propext}, \texttt{Classical.choice}, and \texttt{Quot.sound} for every
audited declaration.

The development formalizes the two-sparse positive-body consequence as its sole
scientific premise and proves all Dynamic MoE implications from it.  The
promoted declarations reuse shared path-cost definitions, but no randomized
rounding statement appears in their dependency footprint.  The pinned source
and theorem-bearing member are checked by SHA-256 before replay.

The companion archive binds the paper, formal sources, exact controls, and
proof identities with a SHA-256 manifest.  From the extracted archive, the
complete replay is:
\begin{verbatim}
python3 verify_all.py
\end{verbatim}
A complete replay builds the promoted Lean declarations, runs 33 focused tests,
regenerates all three canonical JSON outputs byte for byte, and compiles this
paper.  The controls cover the positive-body reduction and deterministic
rounding, including a 200-round oscillation and a false rule without hysteresis
that demonstrates why the dead band is necessary.  The manifest records the
exact case counts, evidence hashes, and promoted proof hashes.

\section{Relation to online rounding and open directions}

The main conclusion is that $\CRdet(k)\leq10\CPB$ for every replica budget
$k\geq1$.  Because $\CPB$ is absolute, this is an $O(1)$ upper bound.  Thus the
deterministic ratio has the same asymptotic order as the randomized ratio from
the previous version, and randomization is not asymptotically necessary.

Vergeres's balanced-descent route is complementary: it gives an explicit
deterministic factor-$16$ fractional algorithm and, after Lazy Threshold
Rounding, a randomized factor-$48$ integral algorithm against an oblivious
adversary~\cite{Vergeres2026}.  Theorem~\ref{thm:det-rounding} instead supplies
the missing deterministic integral conversion.  Composed with the
positive-body path, it removes both random thresholds and the
oblivious-adversary qualification, although we do not claim a polynomial-time
implementation of the complete algorithm.

Correlated sampling on the hypersimplex gives dimension-free randomized
movement couplings.  Naor et al. obtain expected $O(\log k)$
stretch~\cite{NaorEtAl2025}; Anari, Haqi, and Ma give conditional expected
constant stretch under a spectral hypothesis~\cite{AnariEtAl2026}.  Those
results preserve marginals.  Theorem~\ref{thm:det-rounding} instead gives a
deterministic pathwise guarantee for allocations with multiplicity and couples
movement to reciprocal service domination.

The theorem resolves the asymptotic deterministic integral ratio, but not the
best numerical constant.  The concrete integer tangent grid has $O(mk)$
covering rows per event; the geometric grid after Lemma~\ref{lem:tangent}
reduces this to $O(m\log(k+1))$.  The reset schedule uses exact real numbers,
and the formal theorem does not claim a polynomial-time implementation of the
complete algorithm.  Sharper constants and an explicit bit-complexity analysis
remain open.

\paragraph{Research-system disclosure.}
Beyond, the research system operated by Nth Research Collective, materially
assisted with literature retrieval, hypothesis generation, formalization,
executable controls, and adversarial review.  Ian D'Ambrosio selected the
target, directed the investigation, reconciled the evidence, verified the final
claims, and accepts full responsibility for the manuscript.

\end{document}